\documentclass{ifacconf}

\usepackage{amsmath,amssymb,amsfonts}
\usepackage{graphicx}
\usepackage{booktabs}
\usepackage{algorithm}
\usepackage{algpseudocode}
\usepackage{mathtools,color,comment,mathrsfs,url}
\usepackage[short]{optidef}
\usepackage{natbib}

\DeclareMathOperator*{\argmin}{arg\,min}

\newcommand{\diag}{\operatorname{diag}}

\newcommand{\D}{\mathbb D}
\newcommand{\C}{\mathbb C}

\newcommand{\RE}{\mathrm{RE}}
\newcommand{\NP}{\mathrm{NP}}
\newcommand{\supp}{\operatorname{supp}}
\newtheorem{theorem}{Theorem}

\newtheorem{proposition}{Proposition}

\newtheorem{definition}{Definition}

\makeatletter
\@ifundefined{proof}{%
  \newenvironment{proof}[1][Proof]{\par\noindent\textit{#1.}\ }{\hfill$\square$\par}%
}{}
\makeatother

\begin{document}

\begin{frontmatter}

\title{Randomized Atomic Feature Models for Physics-Informed Identification of Dynamic Systems}

\author[First]{Rajiv Singh}
\author[Second]{Mario Sznaier}
\author[Third]{Lennart Ljung}

\address[First]{The MathWorks Inc., 1 Apple Hill Drive, Natick, MA 01760 USA}
\address[Second]{ECE Dept., Northeastern University, Boston, MA 02115, USA}
\address[Third]{Div. of Automatic Control, Link\"{o}ping University, Sweden}

\begin{abstract}
We present a physics-informed framework for system identification based on \emph{randomized stable atomic features}.  Impulse responses are represented as random superpositions of stable atoms, namely damped complex exponentials associated with poles sampled inside a prescribed disk.  Identification is then cast as a convex regularized least-squares problem with optional linear, second-order-cone, and KYP constraints.  The approach generalizes random Fourier and random Laplace features to the damped, nonstationary regime relevant to engineering systems while retaining modal interpretability and scalable finite-dimensional computation.  The main analytic point is an operator-theoretic Disk--Bochner viewpoint: positive measures over stable poles generate positive-definite kernels with a radius-dependent shift defect, while a converse scalar disk moment representation for an arbitrary kernel is characterized by subnormality of the canonical shift.  We prove this statement, establish an RKHS-to-$\ell_1$ embedding, show that sampled poles induce a valid finite atomic gauge, discuss random-feature convergence, and state sparse-recovery guarantees conditionally on the restricted-eigenvalue properties of the realized disk-Vandermonde or input-output design matrix.  We also connect the normalized transfer-function problem to Nevanlinna--Pick interpolation and LFT set-membership.  The framework directly encodes stability margins, modal localization, DC-gain bounds, monotonicity, passivity, relative degree, settling-time targets, and time/frequency-domain error bounds.  Numerical comparisons illustrate how physically meaningful priors can compensate for poor excitation and improve constrained impulse-response recovery in an under-informative data setting.
\end{abstract}


\end{frontmatter}

\section{Introduction}\label{sec:intro}
System identification lies at the intersection of data-driven learning and physical modeling.  Modern methods can be organized along a continuum between \emph{finite-dictionary sparse regression} and \emph{infinite-dimensional kernel regularization}.  At one end, sparse regression techniques, such as the Sparse Identification of Nonlinear Dynamics \citep{BRUNTON2016710} and related dictionary-based formulations, represent dynamics as linear combinations of preselected basis functions.  These approaches are attractive because, once the dictionary is fixed, estimation reduces to linear least squares with sparsity-promoting penalties.  However, the flexibility of such methods is also their burden: physical priors such as stability, passivity, DC-gain bounds, known time constants, resonant frequency bands, monotonicity, or settling-time requirements must be encoded by hand through custom basis functions or additional constraints.  The resulting optimization problem can quickly lose the simplicity that made the original dictionary attractive.

At the other end, kernel-based methods represent impulse responses or nonlinear maps as elements of a reproducing kernel Hilbert space (RKHS) associated with a positive-definite kernel.  For linear system identification, kernels such as tuned-correlated (TC), diagonal-correlated (DC), stable-spline, amplitude-modulated locally stationary (AMLS), and simulation-induced kernels encode stability, smoothness, and correlation decay directly in the regularizer; their hyperparameters can often be selected by empirical Bayes or cross-validation \citep{pillonetto2014kernel,CHEN2018109,chen2019si}.  Kernel methods are powerful because they handle very large or infinite dictionaries implicitly.  The price is that the atoms are hidden, the Gram matrix scales with the number of data points, and it can be difficult to impose several hard engineering constraints simultaneously without redesigning the kernel or solving hybrid constrained problems.

Random feature methods bridge these viewpoints.  By classical Bochner theory, stationary positive-definite kernels can be represented as Fourier transforms of positive measures and approximated by random Fourier features \citep{rahimi2007random}.  Semigroup kernels admit analogous random Laplace features \citep{Yang2014}.  These constructions reduce kernel learning to finite-dimensional linear regression while preserving probabilistic approximation guarantees.  However, Fourier features describe undamped oscillations on the unit circle, and Laplace features describe purely real decays.  Neither geometry by itself naturally captures the stable oscillatory modes that dominate many engineering systems.  A sampled point $p=re^{i\theta}$ inside the unit disk represents both damping, through $r$, and oscillation, through $\theta$.  This observation motivates randomized atomic features over the disk.

A related literature is atomic norm minimization for system identification, where atoms are impulse responses of stable first-order systems or damped sinusoids \citep{Recht2012,Yilmaz2018}.  Randomized active-set and fully corrective Frank--Wolfe methods have shown that sampling candidate poles inside the disk can yield scalable parsimonious identification, including for MIMO systems \citep{MillerATOM:2020}.  The present paper provides a kernel-learning and operator-theoretic viewpoint on such constructions.  Each sampled pole $p$ gives an explicit atom $\phi_p(t)=p^t$.  Randomized sampling approximates an underlying disk-supported moment kernel, while convex penalties and constraints select a low-complexity model consistent with measured data and physical side information.

The geometric relationship is summarized in Fig.~\ref{fig:bochner_triad}.  The Fourier case places mass on the unit circle and produces shift-invariant kernels.  The Laplace case places mass on the positive real axis after a logarithmic change of variables and produces semigroup kernels.  Disk-supported pole measures unify damping and oscillation and produce kernels of the form
\begin{equation}\label{eq:intro-kernel}
K(s,t)=\int_{|p|\le\rho}p^s\overline p^{\,t}\,d\mu(p),
\qquad s,t\in\mathbb N_0 .
\end{equation}
These kernels are positive definite and automatically satisfy a radius-dependent shift-defect inequality.  A central nuance is that the converse is not true under shift-contractivity alone.  A scalar disk moment representation requires the canonical shift associated with the kernel to be subnormal.  This distinction is important because a radius defect gives contractivity of the canonical shift, whereas a scalar pole-measure representation requires the stronger spectral structure supplied by subnormality.

\begin{figure}
\centering
\includegraphics[width=0.99\linewidth]{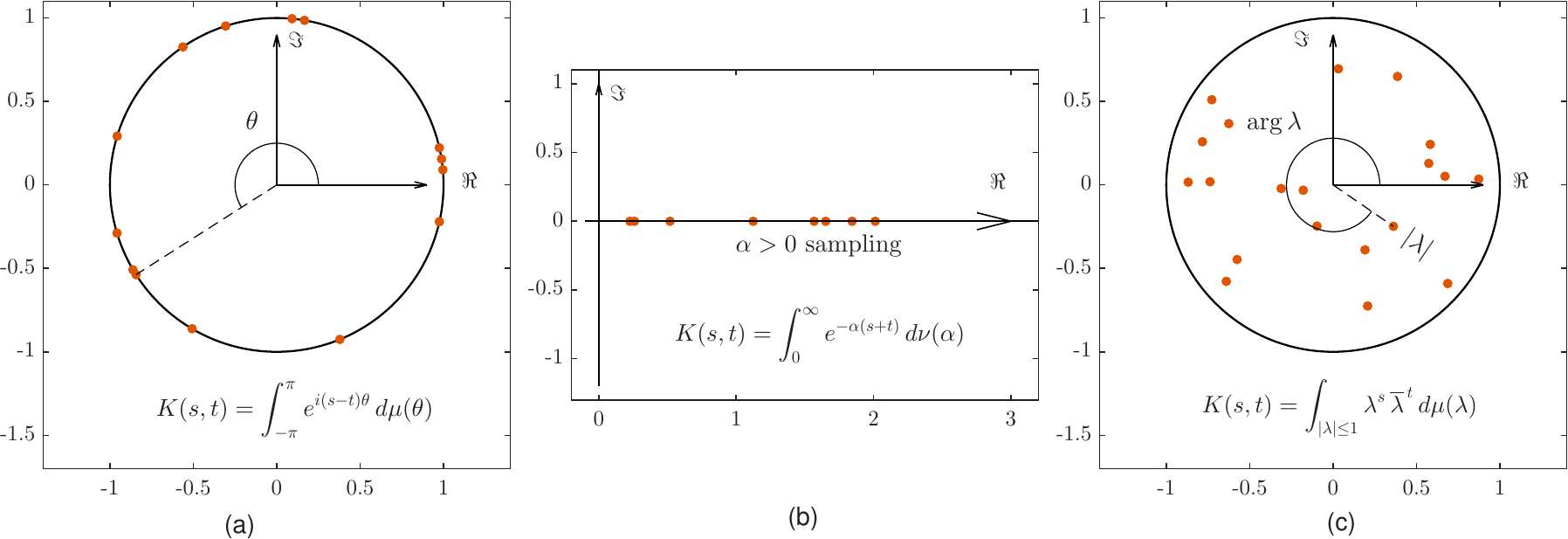}
\caption{Unified view of kernel representations. (a) Fourier/shift-invariant kernels on the unit circle. (b) Laplace/semigroup kernels on the real axis. (c) Disk-supported kernels generated by stable atomic features. Each construction samples a different complex-domain geometry.}
\label{fig:bochner_triad}
\end{figure}

\subsection*{Contributions}
The paper makes five contributions.  First, it unifies finite modal dictionaries, kernel regularization, and randomized atomic features through disk-supported pole measures.  Second, it states and proves a Disk--Bochner theorem for subnormal kernel shifts: positive disk measures imply positive definiteness and a radius defect, and the converse scalar disk moment representation is characterized by subnormality of the canonical shift.  Third, it develops the Nevanlinna--Pick connection in a consistent disk-analytic coordinate: NP positivity gives an exact Schur set-membership and LFT parametrization for gain-normalized transfer functions, while radius constraints are handled separately through pole restrictions or kernel radius defects.  Fourth, it gives stability, approximation, atomic-gauge, and sparse-recovery statements appropriate for RAF models, including the role of radial as well as angular separation in disk-Vandermonde conditioning.  Fifth, it shows how many engineering priors can be imposed by sampling laws or convex constraints in a single master problem.

\section{Randomized Stable Atomic Features}\label{sec:RAF}
Consider discrete-time SISO data $\{(u(t),y(t))\}_{t=0}^{T-1}$ generated approximately by
\begin{equation}\label{eq:conv-model}
y(t)=\sum_{\tau\ge0}h(\tau)u(t-\tau)+e(t),
\end{equation}
where $h$ is a stable impulse response and $e$ represents noise, unmodeled dynamics, or both.  The RAF model approximates $h$ by
\begin{equation}\label{eq:model}
h(t)=\sum_{m=1}^M c_m p_m^t,
\qquad |p_m|\le\rho<1,
\end{equation}
where the pole locations $p_m$ are sampled from a prescribed distribution over an admissible disk or sector, and the residues $c_m$ are learned from data.  For real-valued impulse responses, complex poles and residues are paired with their conjugates.

For direct impulse-response fitting, one forms the Vandermonde matrix
\begin{equation}\label{eq:Phi}
\Phi_{t,m}=p_m^t,
\qquad t=0,\ldots,T-1,
\end{equation}
and estimates $c$ from impulse-response samples.  For input-output data, the relevant design matrix is instead
\begin{equation}\label{eq:Xconv}
X_{t,m}=(\phi_m*u)(t),
\qquad \phi_m(t)=p_m^t,
\end{equation}
so that $y_{\rm model}=Xc+Du$ when a feedthrough term is included.  The basic estimator is
\begin{equation}\label{eq:raf-opt}
(\widehat c,\widehat D)\in\argmin_{c,D} \frac12\|y-Xc-Du\|_2^2+\lambda\Omega(c),
\end{equation}
where $D$ may be fixed or optimized and $\Omega$ may be an $\ell_1$ penalty, a group penalty over conjugate pairs, a channel-structured penalty, or a reweighted sparsity surrogate.  The key point is that the pole locations are fixed during the convex stage.  Nonconvex pole refinement can be applied later, but the main identification problem remains a regularized convex regression with interpretable columns.

\subsection{Atomic gauge induced by sampled poles}\label{subsec:atomic-gauge}
The use of the term ``atomic'' is justified by the gauge generated by the pole atoms.  For a compact admissible pole set $\Omega_p\subset\{p\in\C:|p|\le\rho<1\}$ and a finite horizon $T$, define
\begin{equation}\label{eq:finite-atom}
a_T(p)=[1\;p\;\cdots\;p^{T-1}]^\top,
\qquad
\widetilde a_T(p)=\frac{a_T(p)}{\eta_T(p)},
\end{equation}
where one may take $\eta_T(p)=\|a_T(p)\|_2$.  Other positive normalizations are possible and simply encode different radial weights.  The complex-balanced atomic set is
\begin{equation}\label{eq:disk-atomic-set}
\mathcal A_T(\Omega_p)=
\left\{e^{i\theta}\widetilde a_T(p):p\in\Omega_p,\;\theta\in[0,2\pi)\right\}.
\end{equation}
Its Minkowski functional is
\begin{equation}\label{eq:continuous-gauge}
\gamma_{\Omega_p,T}(h)=
\inf\left\{\sum_k |d_k|:
 h=\sum_k d_k\widetilde a_T(p_k),\;p_k\in\Omega_p\right\}.
\end{equation}

\begin{proposition}[Validity of the disk atomic gauge]\label{prop:gauge}
$\gamma_{\Omega_p,T}$ is a closed convex gauge and is a norm on $\operatorname{span}\mathcal A_T(\Omega_p)$.  For sampled poles $p_1,\ldots,p_M$, the finite RAF gauge
\begin{equation}\label{eq:finite-gauge}
\gamma_{M,T}(h)=
\inf\left\{\sum_{m=1}^M |d_m|:
 h=\sum_{m=1}^M d_m\widetilde a_T(p_m)\right\}
\end{equation}
is likewise a closed convex gauge and a norm on the span of the sampled atoms.
\end{proposition}
\begin{proof}
The map $p\mapsto\widetilde a_T(p)$ is continuous and $\Omega_p$ is compact, hence $\mathcal A_T(\Omega_p)$ is compact.  Because all phases are included, the set is balanced.  Therefore $\operatorname{conv}\mathcal A_T(\Omega_p)$ is compact, convex, balanced, and contains the origin.  Its Minkowski functional is a closed convex gauge.  Restricted to the linear span of the atoms it is finite and positive definite, hence a norm.  The finite sampled case is the same argument with a finite compact set.
\end{proof}

Thus, after the sampled poles are fixed, the $\ell_1$ penalty in \eqref{eq:raf-opt} is exactly a finite atomic-gauge penalty.  Since the sampled atom set is a subset of the continuous disk atom set, $\gamma_{M,T}$ is an inner approximation of $\gamma_{\Omega_p,T}$ and generally satisfies $\gamma_{M,T}(h)\ge\gamma_{\Omega_p,T}(h)$ on their common domain.  Dense resampling improves this approximation, but equality with the continuous atomic norm is not automatic.  If unnormalized columns $a_T(p_m)$ are used in the solver, the same statement holds with the corresponding radius-dependent weights absorbed into the coefficients; column normalization is recommended when sparsity weights should not depend implicitly on atom energy.

The random-feature interpretation arises by drawing $p_m\sim\mu$ and considering the empirical kernel
\begin{equation}\label{eq:empirical-kernel}
\widehat K_M(s,t)=\frac1M\sum_{m=1}^M p_m^s\overline p_m^{\,t}.
\end{equation}
This is an unbiased approximation of the disk-supported kernel 
\begin{equation}\label{eq:Kmu-short}
K_\mu(s,t)=\int p^s\overline p^{\,t}\,d\mu(p).
\end{equation}
Thus RAFs are explicit Monte Carlo features for a stable modal RKHS.  Unlike a purely kernelized estimator, the columns of $X$ have direct physical meaning: each column is the response of a stable first-order mode to the measured input.  This makes it natural to impose modal priors through the sampling law, such as annular regions for time constants or sectors for resonant bands, and to impose response-level priors through convex constraints on $c$ and $D$.

\subsection{Relationship with AMLS kernels}
Writing $p=re^{i\theta}$ and setting $\alpha=-\log r$, $\omega=\theta$ gives
\begin{equation}\label{eq:Knu}
K_\nu(s,t)=\iint e^{-\alpha(s+t)}e^{i\omega(s-t)}\,d\nu(\alpha,\omega),
\end{equation}
where $\nu$ is the pushforward of $\mu$ under $p\mapsto(\alpha,\omega)$.  This is a continuous mixture of locally stationary kernels whose stationary part depends on $s-t$ and whose amplitude modulation depends on $s+t$.  In this sense, disk-supported pole mixtures provide an explicit random-feature representation for an AMLS-like class of kernels \citep{CHEN2018109}.  The disk view is useful because it exposes the pole geometry directly, while the AMLS view connects the construction to existing kernel design methods for regularized system identification.

\section{Kernel Properties and Random-Feature Approximation}\label{sec:kernel-props}
This section records the basic properties of disk-supported RAF kernels.  The proofs are included because the details clarify which statements are consequences of positive measures and which require additional operator-theoretic structure.

\begin{theorem}[PSD and strict PD]\label{thm:psd}
For any finite positive Borel measure $\mu$ on $\overline\D_\rho$, the kernel
\begin{equation}\label{eq:Kmu}
K_\mu(s,t)=\int_{|p|\le\rho} p^s\overline p^{\,t}\,d\mu(p)
\end{equation}
is Hermitian positive semidefinite.  If $\supp(\mu)$ has a limit point in $\D_\rho$ and the time indices are distinct, then $K_\mu$ is strictly positive definite on every finite set of such indices.
\end{theorem}
\begin{proof}
For finite coefficients $a_i$ at times $t_i$,
\begin{align*}
\sum_{i,j}a_i\overline a_jK_\mu(t_i,t_j)
&=\int_{|p|\le\rho}\sum_{i,j}a_i\overline a_jp^{t_i}\overline p^{\,t_j}\,d\mu(p)\\
&=\int_{|p|\le\rho}\left|\sum_i a_ip^{t_i}\right|^2d\mu(p)\ge0.
\end{align*}
If the quadratic form is zero, the polynomial $f(p)=\sum_i a_ip^{t_i}$ vanishes $\mu$-almost everywhere.  If the support has an accumulation point in the open disk, the identity theorem for analytic functions forces $f\equiv0$.  Since monomials with distinct exponents are linearly independent, all $a_i=0$.
\end{proof}

\begin{theorem}[Uniform convergence of RAF kernels]\label{thm:uniform}
Let $\mu$ be a probability measure supported on $|p|\le\rho\le1$ and let $p_1,\ldots,p_M$ be i.i.d. samples from $\mu$.  For any finite horizon $T$ and $\varepsilon>0$,
\begin{align}
&\Pr\left(
\max_{0\le s,t\le T}
|\widehat K_M(s,t)-K_\mu(s,t)|>\varepsilon
\right)\nonumber\\
&\qquad\le 4(T+1)^2\exp\left(-M\varepsilon^2/8\right).
\end{align}
\end{theorem}
\begin{proof}
For every fixed pair $(s,t)$, the random variable $p^s\bar p^{\,t}$ is complex-valued and has magnitude at most one when $\rho\le1$.  Apply Hoeffding's inequality separately to the real and imaginary parts and then use a union bound over the $(T+1)^2$ entries.  Constants are not optimized; the statement is intended as a finite-horizon Monte Carlo guarantee analogous to random Fourier and random Laplace feature bounds \citep{rahimi2007random,Yang2014}.
\end{proof}

\begin{theorem}[RKHS-to-$\ell_1$ embedding]\label{thm:l1}
Let $K_\mu$ be the disk kernel in \eqref{eq:Kmu}, where $\mu$ is a finite positive Borel measure over pole locations.  If $\sup\{|p|:p\in\supp(\mu)\}=\rho<1$ and $\|\mu\|<\infty$, then every $h\in\mathcal H(K_\mu)$ satisfies
\begin{equation}\label{eq:l1-embedding}
\sum_{t=0}^{\infty}|h(t)|\le\frac{\sqrt{\|\mu\|}}{1-\rho}\|h\|_{\mathcal H(K_\mu)}.
\end{equation}
Consequently, functions in $\mathcal H(K_\mu)$ are BIBO-stable impulse responses.
\end{theorem}
\begin{proof}
By the reproducing property and Cauchy--Schwarz,
\begin{equation*}
|h(t)|\le\|h\|_{\mathcal H(K_\mu)}\sqrt{K_\mu(t,t)}.
\end{equation*}
Since
\begin{equation*}
K_\mu(t,t)=\int |p|^{2t}d\mu(p)\le\|\mu\|\rho^{2t},
\end{equation*}
summing the resulting geometric bound over $t$ gives \eqref{eq:l1-embedding}.  Here $\mu$ is the pole measure that defines $K_\mu$ in \eqref{eq:Kmu}.  The proof uses only its total mass $\|\mu\|$ and its support radius $\rho$; it does not require $\mu$ to be discrete.  Thus the same bound applies to finite atomic measures, such as $\mu=\sum_j w_j\delta_{p_j}$ arising from sampled RAF dictionaries, and to non-atomic or continuous pole distributions used to define limiting kernels.
\end{proof}

\section{A Disk--Bochner Theorem with Subnormal Kernel Shifts}\label{sec:diskbochner}
The name ``Disk--Bochner'' is meant to emphasize an analogy with classical Bochner and Herglotz representations: positive kernels can sometimes be represented as moments of positive measures on a spectral domain.  For RAFs the spectral domain is a disk.  The forward direction is elementary and is exactly what the algorithm uses.  The converse requires care.

\begin{definition}[Canonical Hilbert space and shift]\label{def:canonical}
Let $K$ be a Hermitian PSD kernel on $\mathbb N_0$.  Introduce formal vectors $e_0,e_1,\ldots$ with inner product $\langle e_s,e_t\rangle=K(s,t)$, quotient by zero-norm vectors, and complete to obtain $\mathcal H_K$.  The canonical shift is the operator initially defined on the linear span of the $e_t$ by
\begin{equation}
Te_t=e_{t+1}.
\end{equation}
When this map is well defined and bounded, the pair $(T,e_0)$ is the canonical cyclic realization of $K$ because
\begin{equation}\label{eq:canonical-moments}
K(s,t)=\langle T^se_0,T^te_0\rangle_{\mathcal H_K}.
\end{equation}
\end{definition}

\begin{definition}[Subnormality]
A bounded operator $T$ on a Hilbert space $\mathcal H$ is subnormal if there exists a larger Hilbert space $\mathcal G\supseteq\mathcal H$ and a normal operator $N$ on $\mathcal G$ such that $\mathcal H$ is invariant under $N$ and $T=N|_{\mathcal H}$.  Normality means $N^*N=NN^*$.
\end{definition}

\begin{proposition}[Disk measures imply a radius defect]\label{prop:forward}
If $K$ has representation \eqref{eq:Kmu} with $\supp(\mu)\subseteq\overline\D_\rho$, then $K$ is PSD and satisfies
\begin{equation}\label{eq:radius-defect}
\rho^2K(s,t)-K(s+1,t+1)\succeq0.
\end{equation}
Moreover, the canonical shift is subnormal.
\end{proposition}
\begin{proof}
PSD follows from Theorem~\ref{thm:psd}.  For the radius defect,
\begin{align*}
&\sum_{s,t}a_s\overline a_t\big(\rho^2K(s,t)-K(s+1,t+1)\big)\\
&\qquad=\int_{|p|\le\rho}(\rho^2-|p|^2)\left|\sum_s a_sp^s\right|^2d\mu(p)\ge0.
\end{align*}
For subnormality, let $N$ be multiplication by $p$ on $L^2(\mu)$.  This is a normal operator with spectrum contained in $\overline\D_\rho$.  The closed subspace generated by $1,p,p^2,\ldots$ is invariant under $N$, and the restriction has moments $K$.  Hence the canonical shift is unitarily equivalent to a subnormal restriction of $N$.
\end{proof}

\begin{proposition}[Radius defect and contractivity]\label{prop:defect-contractivity}
For a Hermitian PSD kernel $K$, the radius defect \eqref{eq:radius-defect} is equivalent to a bounded canonical shift satisfying $\|T\|\le\rho$.
\end{proposition}
\begin{proof}
For $x=\sum_s a_se_s$,
\begin{equation}
\rho^2\|x\|^2-\|Tx\|^2=\sum_{s,t}a_s\overline a_t\big(\rho^2K(s,t)-K(s+1,t+1)\big).
\end{equation}
Thus the defect is PSD precisely when $\|Tx\|\le\rho\|x\|$ on the dense span of the canonical vectors, and hence on $\mathcal H_K$ by continuity.
\end{proof}

\begin{theorem}[Disk--Bochner]\label{thm:db}
Let $K$ be a Hermitian PSD kernel on $\mathbb N_0$.  Then $K$ admits
\begin{equation}\label{eq:db-main}
K(s,t)=\int_{|p|\le\rho}p^s\overline p^{\,t}\,d\mu(p)
\end{equation}
for a finite positive Borel measure $\mu$ on $\overline\D_\rho$ if and only if its canonical shift is subnormal and has a normal extension $N$ with $\sigma(N)\subseteq\overline\D_\rho$.  In either case the radius defect \eqref{eq:radius-defect} holds.
\end{theorem}
\begin{proof}
The forward implication follows from Proposition~\ref{prop:forward}.  Conversely, suppose the canonical shift $T$ has a normal extension $N$ with spectral measure $E(\cdot)$ and spectrum in $\overline\D_\rho$.  Let $b=e_0$.  Since $T^sb=N^sb$ for all $s\ge0$,
\begin{equation}
K(s,t)=\langle N^sb,N^tb\rangle.
\end{equation}
By the spectral theorem,
\begin{equation}
N=\int_{\sigma(N)}p\,dE(p).
\end{equation}
Define the scalar measure $\mu(\Omega)=\langle E(\Omega)b,b\rangle$.  Then $\mu$ is finite, positive, supported in $\overline\D_\rho$, and
\begin{equation}
K(s,t)=\int_{\sigma(N)}p^s\overline p^{\,t}\,d\mu(p).
\end{equation}
\end{proof}

\begin{proposition}[Uniqueness]\label{prop:unique}
If a finite positive measure satisfying \eqref{eq:db-main} exists, it is unique.
\end{proposition}
\begin{proof}
If two measures agree on all mixed moments $p^s\overline p^{\,t}$, their difference integrates to zero against every polynomial in $p$ and $\overline p$.  This algebra contains constants, separates points on the compact disk, and is closed under conjugation.  By Stone--Weierstrass it is dense in $C(\overline\D_\rho)$, so the measures agree on all continuous functions and are equal by the Riesz representation theorem.
\end{proof}

\subsection{Necessity of the subnormality hypothesis}
The one-step defect condition is equivalent to contractivity of the canonical shift.  A general contraction, however, need not be subnormal.  The following elementary example shows why subnormality is the right hypothesis for a scalar disk moment representation.  Let
\begin{equation}
A=\begin{bmatrix}0&1\\0&0\end{bmatrix},
\qquad b=\begin{bmatrix}0\\1\end{bmatrix},
\end{equation}
and define $K(s,t)=\langle A^sb,A^tb\rangle$.  Then $K$ is a Gram kernel and is PSD.  Since $A$ is a contraction, the unit-radius defect $K(s,t)-K(s+1,t+1)$ is also PSD.  Its diagonal moments are
\begin{equation}
K(0,0)=1,
\qquad K(1,1)=1,
\qquad K(k,k)=0 \;\; (k\ge2).
\end{equation}
If a positive disk measure existed, these equations would imply
\begin{equation}
\int1\,d\mu=1,
\qquad \int|p|^2d\mu=1,
\qquad \int|p|^4d\mu=0.
\end{equation}
The last equality forces $p=0$ almost everywhere, contradicting the second.  Thus shift-contractivity alone does not characterize disk moment kernels.

\section{Radius Normalization and the Nevanlinna--Pick Connection}\label{sec:NP}
Nevanlinna--Pick (NP) theory is a statement about Schur-class analytic functions on the unit disk.  To avoid notational ambiguity, this section uses the delay variable
\[
        \zeta=z^{-1},
\]
so that a stable discrete-time transfer function is written as
\begin{equation}\label{eq:delay-transfer}
        H(\zeta)=D+\sum_m\frac{c_m}{1-p_m\zeta},
        \qquad |\zeta|<1 .
\end{equation}
The usual frequency response is recovered by setting \(\zeta=e^{-\mathrm i\omega}\).  Thus the pole parameter \(p_m\) is the discrete-time system pole, while \(\zeta\) is the analytic variable used by Schur and NP theory.

There are two related but distinct normalizations.  The first is the pole-radius normalization of the Disk--Bochner kernel.  If
\[
        K(s,t)=\int_{|p|\le\rho}p^s\bar p^t\,d\mu(p),
        \qquad 0<\rho\le1,
\]
then
\begin{equation}\label{eq:normalized-kernel}
        \widetilde K(s,t)=\frac{K(s,t)}{\gamma^2\rho^{s+t}}
\end{equation}
is, up to the scalar factor \(\gamma^{-2}\), the pole-moment kernel associated with the normalized pole coordinate \(q=p/\rho\).  The corresponding defect identity is
\begin{equation}\label{eq:normalized-defect}
\widetilde K(s,t)-\widetilde K(s+1,t+1)=
\frac{\rho^2K(s,t)-K(s+1,t+1)}{\gamma^2\rho^{s+t+2}} .
\end{equation}
Hence support in \(|p|\le\rho\) implies the radius defect
\begin{equation}\label{eq:radius-defect-np}
        \rho^2K(s,t)-K(s+1,t+1)\succeq0 .
\end{equation}
This is a structural pole-geometry condition.

The second normalization is the gain normalization used for NP interpolation.  If
\[
        \|H\|_{H_\infty}:=\sup_{|\zeta|<1}\|H(\zeta)\|\le\gamma,
\]
then
\begin{equation}\label{eq:Fnormalized}
        F(\zeta)=\gamma^{-1}H(\zeta)
\end{equation}
belongs to the Schur class.  Given interpolation data
\[
        F(\zeta_i)=W_i,\qquad |\zeta_i|<1,
\]
the matrix-valued NP theorem states that a feasible Schur function exists if and only if
\begin{equation}\label{eq:np-matrix}
        \left[\frac{I-W_iW_j^*}{1-\zeta_i\bar\zeta_j}\right]_{i,j}\succeq0 .
\end{equation}
When \eqref{eq:np-matrix} is feasible, all solutions are parameterized by a linear fractional transformation
\begin{equation}\label{eq:np-lft}
        F=\mathcal T_\Theta[Q],\qquad Q\in\mathcal S,
\end{equation}
where \(\Theta\) is determined by the interpolation data and \(Q\) is a free Schur parameter.  In MIMO problems this is the standard Redheffer/LFT parametrization.  Thus NP theory gives an exact set-membership description of all gain-normalized analytic transfer functions satisfying the interpolation data \citep{sarason:1994sub,Rotstein1996,Parrilo1998}.

The pole-radius normalization \(q=p/\rho\) and the gain normalization \(F=H/\gamma\) should not be conflated.  If one rewrites \eqref{eq:delay-transfer} in normalized pole coordinates \(p=\rho q\), then
\[
        H(\zeta)=D+\sum_m\frac{c_m}{1-\rho q_m\zeta}.
\]
Equivalently, the transfer function with poles \(q_m\) is obtained from the scaled argument \(H(\eta/\rho)\).  A Schur certificate for \(\gamma^{-1}H(\eta/\rho)\) is therefore a gain certificate for a scaled-domain function, not automatically the same as the usual unit-circle \(H_\infty\) bound for \(H(\zeta)\).  In the RAF formulation we keep these roles separate: pole-radius priors are imposed by sampling \(|p_m|\le\rho\) and by the radius defect, while gain constraints are imposed directly on the transfer function in the delay variable \(\zeta\) through Pick, KYP, or frequency-grid constraints.

The RAF dictionary is then a finite-dimensional coordinate system for searching inside a Schur set-membership problem.  Define the rational RAF span in the delay variable by
\begin{equation}\label{eq:VM}
\mathcal V_M=
\left\{\gamma^{-1}\left(D+\sum_{m=1}^M\frac{b_m}{1-p_m\zeta}\right):\ |p_m|\le\rho\right\}.
\end{equation}
For a fixed NP data set, let
\begin{equation}\label{eq:FNP}
\mathcal F_{\NP}=\{F\in\mathcal S:\ F(\zeta_i)=W_i,
\ i=1,\ldots,N\}
\end{equation}
and
\begin{equation}\label{eq:FM}
\mathcal F_M=\mathcal F_{\NP}\cap\mathcal V_M,
\end{equation}
possibly with additional finite-dimensional convex side constraints.  Then \(\mathcal F_M\) is a certified inner approximation of the NP feasible set: every RAF model satisfying the normalized Pick or gain constraints is a member of \(\mathcal F_{\NP}\).

\begin{proposition}[RAF inner approximation]\label{prop:raf-np-inner}
Assume nested sampled pole sets dense in a compact subset of the open disk \(|p|<\rho\). Then:
\begin{enumerate}
\item \(\mathcal F_M\subseteq\mathcal F_{\NP}\) for every \(M\).
\item On any finite impulse-response horizon, \(T\) distinct sampled poles span all length-\(T\) impulse-response vectors.
\item In \(H_2\), equivalently in \(\ell_2\) impulse-response norm, the closed span of the atoms \(a(p)=(1,p,p^2,\ldots)\) over any set with an accumulation point in the disk is the whole space.
\item Simple-pole rational members of \(\mathcal F_{\NP}\) whose poles are included in the dictionary are represented exactly.  More general rational or Schur-feasible functions can be approximated in finite-horizon, compact-open, and \(H_2\)-type topologies, provided the approximants are also certified by the imposed Schur/Pick or conic constraints.  No equality with the full \(H_\infty\) Schur class, and no \(H_\infty\)-norm density for a fixed finite dictionary, is claimed.
\end{enumerate}
\end{proposition}
\begin{proof}
The inclusion is immediate from \eqref{eq:FM}.  For the finite-horizon statement, \(T\) distinct poles give a full-rank Vandermonde matrix on samples \(0,\ldots,T-1\), so their span is all of \(\mathbb C^T\).  For the \(H_2\) statement, suppose \(h\in\ell_2\) is orthogonal to \(a(p)\) for all \(p\) in a set with an accumulation point inside the disk.  Then the analytic function
\[
        g(p)=\sum_{t\ge0}\overline{h(t)}p^t
\]
vanishes on a set with an accumulation point; by the identity theorem \(g\equiv0\), hence all coefficients \(h(t)\) vanish.  Therefore the atom span is dense.  Finally, the NP/LFT formula \eqref{eq:np-lft} parametrizes the feasible family by a free Schur parameter.  Rational Schur parameters yield rational feasible transfer functions, and the LFT is continuous on compact subsets where it is nonsingular.  Dense RAF pole sampling can represent simple-pole rational functions exactly when their poles are present; repeated-pole terms require derivative atoms or approximation by nearby simple poles.  Feasibility with respect to the NP/Pick or gain constraints remains a separate certification step.
\end{proof}

The Disk--Bochner pole-moment kernel and the Pick kernel therefore have different roles.  Monte Carlo pole sampling gives convergence of empirical pole-moment kernels to the chosen \(K_\mu\).  Separately, if RAF transfer functions \(F_M\) converge locally uniformly to a Schur function \(F\), then their Pick kernels
\[
        K_{F_M}(\zeta,\xi)=\frac{I-F_M(\zeta)F_M(\xi)^*}{1-\zeta\bar\xi}
\]
converge locally uniformly to \(K_F\) on compact subsets.  The first convergence is a feature-kernel statement; the second is a transfer-function contractivity statement.

\section{Sparse Recovery and Disk-Vandermonde Conditioning}\label{sec:sparse-recovery}
The sparse regression layer of RAF uses standard Lasso or group-Lasso mechanisms, but the statistical guarantees must be stated in terms of the actual realized design matrix.  Let $Z_T$ denote the regression matrix used in the finite data problem: $Z_T=\Phi$ for direct impulse-response fitting and $(Z_T)_{t,m}=(\phi_m*u)(t)$ for input-output fitting.  The familiar sample-size rate is obtained when the empirical Gram matrix $G_T=Z_T^*Z_T/T$ satisfies an $s$-sparse restricted-eigenvalue condition with constant $\kappa_{\RE}>0$ and the columns are scaled so that their empirical norms are $O(1)$.  Under sub-Gaussian noise with variance proxy $\sigma^2$ and a true coefficient vector $c^\star$ supported on a set $S$ of size $s$, standard arguments yield
\begin{equation}\label{eq:lassa-rate}
\|\widehat c-c^\star\|_2
\le C\frac{\sigma}{\kappa_{\RE}}\sqrt{\frac{s\log M}{T}},
\end{equation}
up to constants depending on normalization and group structure \citep{Bickel:2008}.  For raw impulse-response Vandermonde columns with $|p_m|<1$, the column norms may saturate as $T$ grows; in that case one should use the normalized finite-horizon Gram matrix, and the dependence on $T$ is captured by the resulting conditioning constant rather than by the displayed rate alone.  Exact support recovery is stronger and requires additional beta-min and incoherence or irrepresentability assumptions.  Therefore the role of RAF sampling is not to guarantee recovery for every disk dictionary, but to construct a dictionary whose realized design matrix is sufficiently well conditioned for the data and priors at hand.

For the pure disk-Vandermonde matrix $\Phi_{t,m}=p_m^t$, conditioning depends on both angular and radial separation.  The finite-horizon Gram matrix is
\begin{equation}\label{eq:disk-gram}
G_{ij}=\sum_{t=0}^{T-1}\overline p_i^{\,t}p_j^t
=\frac{1-(\overline p_i p_j)^T}{1-\overline p_i p_j}.
\end{equation}
After column normalization, the pairwise coherence is
\begin{equation}\label{eq:muT}
\mu_T(p_i,p_j)=
\frac{\left|\sum_{t=0}^{T-1}(\overline p_i p_j)^t\right|}
{\left(\sum_{t=0}^{T-1}|p_i|^{2t}\right)^{1/2}
 \left(\sum_{t=0}^{T-1}|p_j|^{2t}\right)^{1/2}}.
\end{equation}
A simple sufficient condition for a normalized sub-Gram matrix is
\begin{equation}\label{eq:coh-suff}
\max_{i\ne j\in S}\mu_T(p_i,p_j)\le\mu_s
\quad\Rightarrow\quad
\lambda_{\min}(G_S)\ge1-(s-1)\mu_s,
\end{equation}
which follows from Gershgorin's theorem.  This is conservative, but it highlights that angular separation alone is insufficient for disk-supported atoms.  Two poles with close radii and close angles can be nearly collinear over a finite horizon, and poles near the unit circle can require longer horizons to separate.

The infinite-horizon limit exposes the natural geometry.  For $|p_i|,|p_j|<1$,
\begin{equation}\label{eq:muinf}
\mu_\infty(p_i,p_j)=
\frac{\sqrt{(1-|p_i|^2)(1-|p_j|^2)}}{|1-\overline p_i p_j|},
\end{equation}
and
\begin{equation}\label{eq:pseudohyp}
1-\mu_\infty(p_i,p_j)^2=
\left|\frac{p_i-p_j}{1-\overline p_i p_j}\right|^2.
\end{equation}
The right-hand side is the squared pseudo-hyperbolic distance in the disk, familiar from Hardy-space theory.  For poles with the same angle, $p_i=r_ie^{i\theta}$ and $p_j=r_je^{i\theta}$,
\begin{equation}
d_{\D}(p_i,p_j)=\frac{|r_i-r_j|}{|1-r_ir_j|}.
\end{equation}
Radial separation is therefore part of the conditioning story.  In this paper, a ``well-spread'' pole distribution should be understood as one that produces small finite-horizon coherence, or more generally a favorable RE constant, for the candidate supports encountered by the estimator.  Developing sharp RIP/RE theory for disk-Vandermonde dictionaries is an interesting problem but is outside the scope of this work.

\section{Engineering Priors as Sampling Laws and Convex Constraints}\label{sec:priors}
The practical advantage of the RAF parameterization is that physical knowledge can be enforced either before optimization, by shaping the pole sampling law, or during optimization, by convex constraints on residues and feedthrough terms.  We describe common priors for the SISO case; MIMO extensions use the same parallel realization with block residues and channel-wise group structure.

\subsection{Order and sparsity}
Parsimonious models are encouraged by $\ell_1$ or group-$\ell_1$ penalties over conjugate pole pairs,
\begin{equation}
\sum_{g\in\mathcal G}\|c_g\|_2\le\tau,
\end{equation}
or by adding $\lambda_1\sum_g\|c_g\|_2$ to the objective.  Reweighted variants can be used after an initial solve to reduce bias and sharpen the active set.  Sparsity is primarily a model-selection device: it identifies a small set of active modes from a much larger sampled dictionary.  These active modes should be interpreted as candidate pole regions rather than certified continuous pole estimates; they can guide local resampling or nonlinear refinement when the data, separation, signal-to-noise ratio, and design conditioning are favorable.

\subsection{Stability, settling, and tail budgets}
A stability margin is imposed by sampling only poles with $|p_m|\le\rho<1$.  This is a hard prior; no solver constraint is needed.  For pointwise impulse-response settling after time $T_s$,
\begin{equation}\label{eq:settle}
\sup_{t\ge T_s}|h(t)|
\le\sum_{m=1}^M|c_m||p_m|^{T_s}
\le\epsilon_h.
\end{equation}
This mode-resolved envelope is tighter than the common-radius bound $\rho^{T_s}\sum_m|c_m|$ when the sampled poles have different radii.  For truncation error or bounded-input robustness, the relevant quantity is the $\ell_1$ tail,
\begin{equation}\label{eq:l1-tail}
\sum_{t=T_s}^{\infty}|h(t)|
\le\sum_{m=1}^M\frac{|c_m||p_m|^{T_s}}{1-|p_m|}.
\end{equation}
The global BIBO budget is the case $T_s=0$:
\begin{equation}\label{eq:bibo-budget}
\sum_{m=1}^M\frac{|c_m|}{1-|p_m|}\le h_{\max}.
\end{equation}
Similarly, the denominator $|1-p_m|$ belongs to step-response and DC-gain quantities.  If $s(t)=\sum_{\tau=0}^t h(\tau)$, then
\begin{equation}\label{eq:step-tail}
|s_\infty-s(t)|\le\sum_m\frac{|c_m||p_m|^{t+1}}{|1-p_m|}.
\end{equation}
For complex residues, constraints involving $|c_m|$ are implemented with epigraph variables $a_m\ge |c_m|$, equivalently $\sqrt{(\Re c_m)^2+(\Im c_m)^2}\le a_m$, so the resulting inequalities are SOCP-representable.

\subsection{Modal localization}
Known time constants can be encoded by sampling radii near $r=e^{-\alpha_i}$ or by imposing annular regions $r\in[r_i^{\min},r_i^{\max}]$.  Resonant bands are encoded by restricting angles $\theta_m\in\Theta=\cup_b[\omega_b^{\min},\omega_b^{\max}]$.  Soft preferences can be added through weighted penalties or through sampling distributions concentrated near expected modes.  Hard exclusions, such as rejecting unstable or nonphysical modes, are enforced by simply not sampling those regions.

\subsection{DC gain and static constraints}
For
\begin{equation}
G(z)=D+\sum_m\frac{c_m}{1-p_mz^{-1}},
\end{equation}
the DC gain is
\begin{equation}\label{eq:dcgain}
G(1)=D+\sum_m\frac{c_m}{1-p_m}.
\end{equation}
A bound $|G(1)|\le G_{\max}$ is an SOCP constraint after splitting real and imaginary parts; a known DC gain is imposed by equality.  Similar affine constraints can encode known Markov parameters, zero steady-state error under a specified input class, or sign information at selected frequencies.

\subsection{Passivity}
A finite RAF model has the parallel realization
\begin{align}\label{eq:parallel-realization}
A&=\diag(p_1,\ldots,p_M),&
B&=\mathbf 1,\nonumber\\
C&=[c_1\ \cdots\ c_M],&
D&\in\mathbb R .
\end{align}
Discrete-time positive realness can be imposed by the positive-real KYP LMI \citep{Boyd:1994}.  In the RAF model, $A$ and $B$ are fixed by the sampled poles, while $C$ and $D$ are affine in the decision variables.  The KYP condition is therefore convex in the unknowns together with the storage matrix $P\succ0$.

\subsection{Monotonicity, no overshoot, and relative degree}
For step monotonicity, restrict to real poles $p_m\in(0,1)$ and impose $c_m\ge0$.  Then $h(t)\ge0$ and the step response $s(n)=\sum_{t=0}^nh(t)$ is nondecreasing.  Discrete concavity or no-overshoot refinements can be written as linear inequalities because
\begin{equation}
\Delta(p^t)=p^t-p^{t-1}=-(1-p)p^{t-1}.
\end{equation}
Relative degree $r_d$ is enforced by
\begin{equation}\label{eq:reldeg}
\sum_{m=1}^M c_mp_m^k=0,
\qquad k=0,\ldots,r_d-1.
\end{equation}
Equivalently, shifted features $\phi_m^{(r_d)}(t)=p_m^{t-r_d}\mathbf 1\{t\ge r_d\}$ can be used so that the first $r_d$ taps vanish by construction.

\subsection{Time-domain error bounds}
Given measured output $y_{\rm meas}$ and model output $y_{\rm model}=Xc+Du$, elementwise tolerance bounds are
\begin{equation}\label{eq:timebox}
-\varepsilon(t)\le y_{\rm meas}(t)-(Xc+Du)_t\le\varepsilon(t),
\qquad t=0,\ldots,T-1.
\end{equation}
These are linear inequalities.  Windowed RMS bounds over segments can be imposed as
\begin{equation}\label{eq:window-bound}
\|W_\ell(y_{\rm meas}-Xc-Du)\|_2\le\eta_\ell,
\qquad \ell=1,\ldots,K_{\rm win},
\end{equation}
which are SOCP constraints.  Such constraints are useful in set-membership identification, where the model must fit the data within known time-domain uncertainty envelopes rather than merely minimize average squared error.

\subsection{Frequency response masks}\label{sec:freqmask}
At a grid $\omega_k\in[0,\pi]$,
\begin{equation}\label{eq:Gk}
G_k=D+\sum_{m=1}^M\frac{c_m}{1-p_me^{-i\omega_k}}
\end{equation}
is affine in $(c,D)$.  For MIMO systems, the exact constraint $\sigma_{\max}(G_k)\le\gamma_k$ is equivalent to
\begin{equation}\label{eq:gain-lmi}
\begin{bmatrix}
\gamma_kI&G_k\\
G_k^*&\gamma_kI
\end{bmatrix}\succeq0.
\end{equation}
Directional SOCP approximations use fixed input directions $v_\ell$:
\begin{equation}\label{eq:dir-bound}
\|G_kv_\ell\|_2\le\gamma_k,
\qquad \ell=1,\ldots,L_{\rm dir}.
\end{equation}
Input/output weights are included by replacing $G_k$ with $W_{\rm out}(\omega_k)G_kW_{\rm in}(\omega_k)$.  These constraints provide grid-based frequency shaping within the same convex RAF estimation problem.  A true full-band \(H_\infty\) certificate requires an additional bounded-real/KYP or equivalent robust-control condition rather than only a finite frequency grid.

\subsection{Master problem}\label{sec:master}
Combining these pieces gives a representative convex program.  Let
$y^\mathcal A_{t,m}=(\phi_m*u)(t)$ and
$y_{\rm model}=y^\mathcal Ac+Du$.  The objective is
\begin{equation}\label{eq:master-obj}
\min_{c,D}\; \frac12\|y_{\rm meas}-y^\mathcal Ac-Du\|_2^2
+\lambda_2\|c\|_2^2+
\lambda_1\sum_{g\in\mathcal G}\|c_g\|_2 .
\end{equation}
Representative optional constraints include
{
\begin{align}
-\varepsilon &\le y_{\rm meas}-y^\mathcal Ac-Du\le\varepsilon,\label{eq:master-time}\\
\|W_k(y_{\rm meas}-y^\mathcal Ac-Du)\|_2&\le\eta_k,
\quad k=1,\ldots,K_{\rm win},\label{eq:master-window}\\
G_k&=D+\sum_{m=1}^M\frac{c_m}{1-p_me^{-i\omega_k}},
\quad k=1,\ldots,K_\omega,\label{eq:master-gk}\\
\|G_kv_\ell\|_2&\le\gamma_k,
\quad \ell=1,\ldots,L_{\rm dir},\label{eq:master-dir}\\
\sum_{m=1}^M a_m|p_m|^{T_s}&\le\epsilon_h,\label{eq:master-settle}\\
\sum_{m=1}^M\frac{a_m}{1-|p_m|}&\le h_{\max},\label{eq:master-bibo}\\
\sum_{m=1}^M c_mp_m^k&=0,
\quad k=0,\ldots,r_d-1,\label{eq:master-rd}\\
G(1)&=D+\sum_{m=1}^M\frac{c_m}{1-p_m},\label{eq:master-g1}\\
\left|G(1)\right|&\le G_{\max}.\label{eq:master-dc}
\end{align}}
Here $a_m\ge |c_m|$ are epigraph variables, implemented as second-order cones for complex residues.  Monotone real-pole models add $c_m\ge0$ for $p_m\in(0,1)$.  All constraints above are linear, SOCP, or LMI representable after splitting complex variables into real and imaginary parts.  Hard pole priors do not appear as constraints because they are already encoded in the sampled dictionary.

\begin{algorithm}[t]
\caption{Constrained Randomized Atomic Identification}
\label{alg:rsa}
\begin{algorithmic}[1]
\State \textbf{Inputs:} data $(u,y)$; admissible pole region $\mathcal P$; priors; number of atoms $M$.
\State Sample $p_m$ i.i.d. in $\mathcal P$ and pair conjugates for real outputs.
\State Form the convolved design $y^\mathcal A_{t,m}=(\phi_m*u)(t)$ and the parallel realization with $A=\diag(p_m)$.
\State Solve the convex master problem with selected QP/SOCP/LMI constraints.
\State Prune active groups and optionally resample or refine locally around active poles.
\State \textbf{Return:} $\widehat h(t)=\sum_m\widehat c_mp_m^t$ and $\widehat G(z)=\widehat D+\sum_m\widehat c_m/(1-p_mz^{-1})$.
\end{algorithmic}
\end{algorithm}

\subsection{Pole localization and refinement}
The convex RAF stage chooses residues over a fixed sampled pole dictionary and therefore does not solve the continuous nonlinear pole-estimation problem.  The active RAF atoms should be interpreted as candidate pole regions, not as certified continuous pole estimates.  Pole localization from a finite sampled pole dictionary requires additional assumptions: informative data, sufficiently separated true poles, low noise, adequate signal strength, and favorable conditioning of the normalized input-filtered design.  Dictionary density alone is insufficient, since nearby disk-pole atoms become highly coherent over any finite horizon; a sparse solver may select one of several nearly equivalent atoms, or a small cluster of atoms, without uniquely identifying the true pole.

In practice, we use the sparse RAF solution as a convex initialization.  Active atoms are clustered in the pole plane, preferably using a disk-aware metric such as the pseudo-hyperbolic distance in~\eqref{eq:pseudohyp}.  Cluster representatives then initialize a local nonlinear pole/residue fit, or a prediction-error refinement, subject to the same stability and prior constraints.  The refined model should be validated by simulation error, frequency-domain fit, residual checks, and model-order diagnostics.  This two-stage strategy preserves convexity in the main selection step while using the active support only as a warm start for continuous refinement; exact model-order or pole recovery is not guaranteed by dictionary density alone.

\section{Computational and Modeling Considerations}\label{sec:computational}
This section collects practical considerations. The RAF estimator has three distinct approximation layers.  The first is the modeling approximation: the true impulse response is approximated by a superposition of stable first-order atoms.  The second is the random-feature approximation: a finite sampled dictionary approximates a continuous pole measure or atomic set.  The third is the statistical approximation: the coefficients are estimated from finite, noisy, and possibly poorly exciting data.  These layers should be kept conceptually separate.  Increasing the number of sampled atoms improves dictionary resolution but can worsen conditioning.  Increasing the data length improves statistical estimation but may not help identify modes that are not excited by the input.  Adding physical constraints can reduce variance and rule out spurious explanations, but overly restrictive priors can bias the result.

\subsection{Choice of sampling law and number of atoms}
The sampling law is the primary mechanism for encoding qualitative model knowledge.  A uniform distribution over the disk is rarely the best engineering choice.  For lightly damped mechanical or electrical systems, radii near the stability boundary and angles in expected resonance bands are often more important than strongly damped modes.  For thermal or diffusion-dominated systems, real positive poles and logarithmically spaced time constants may be more appropriate.  For systems with known bandwidth, the angular support can be restricted before optimization, reducing both computation and false discoveries.

The number of atoms $M$ controls a bias-variance-conditioning trade-off.  Small $M$ gives fast, well-conditioned regressions but may miss relevant poles.  Large $M$ improves coverage but creates coherent columns, especially when poles are close in pseudo-hyperbolic distance or when the input signal does not sufficiently distinguish their responses.  In practice, column normalization, group penalties, and iterative resampling are important.  One can start with a moderate global dictionary, solve the convex problem, retain active or nearly active regions, and resample locally.  This produces a multiresolution search over the disk without ever solving a large nonconvex pole-placement problem.

\subsection{Complexity}
For SISO data with $N$ samples and $M$ atoms, forming the convolved design matrix costs $O(NM)$ using direct recursions
\begin{equation}
x_m(t+1)=p_mx_m(t)+u(t+1),
\qquad X_{t,m}=x_m(t),
\end{equation}
up to conventions on initial conditions and delays.  The basic ridge or Lasso problem then scales like a standard dense regression in $M$ variables.  Additional SOCP constraints for epigraph variables, impulse settling, BIBO budgets, time-domain windows, and directional frequency bounds scale linearly with the number of atoms and grid points.  Exact MIMO frequency-gain constraints introduce LMIs whose size is determined by input/output dimension; directional SOCP relaxations are often preferable when the number of grid points is large.  Passivity constraints introduce a storage matrix whose dimension equals the number of active states in the parallel realization, so it may be advantageous to impose passivity after an initial sparsifying solve or on a pruned dictionary.

\subsection{Comparison with related methods}
RAFs should not be viewed as replacing kernel regularization or atomic norm methods; rather, they occupy a useful middle ground.  Compared with kernel regularization, RAFs make the features explicit and allow direct coefficient constraints, but they require a sampled dictionary and can be sensitive to dictionary coherence.  Compared with gridless atomic norm formulations, RAFs trade exact continuous convex geometry for a scalable finite problem that can incorporate many engineering constraints.  Compared with classical PEM, RAFs are less tied to a fixed parametric order and can exploit convex priors before nonlinear refinement.  Table~\ref{tab:method-comparison} summarizes this positioning.

\begin{table}[t]
\centering
\caption{Qualitative comparison of identification approaches.}
\label{tab:method-comparison}
\resizebox{\linewidth}{!}{%
\begin{tabular}{@{}lcccc@{}}
\toprule
Property & PEM & KRM & ANM & RAF \\
\midrule
Explicit modes & yes & no & yes & yes \\
Convex main fit & local & yes & yes & yes \\
Kernel/RKHS view & no & yes & partial & yes \\
Continuous pole search & yes & implicit & yes & refinement \\
Hard side constraints & limited & moderate & moderate & broad \\
Random-feature scalability & no & no & partial & yes \\
\bottomrule
\end{tabular}}
\end{table}

\subsection{MIMO and structured extensions}
For MIMO systems, the sampled pole set can be shared across channels while residues become matrices.  A common-pole model has
\begin{equation}
G(z)=D+\sum_{m=1}^M \frac{C_m}{1-p_mz^{-1}},
\end{equation}
where $C_m$ is an output-by-input residue matrix.  Group penalties on $\|C_m\|_F$ promote a small number of shared modes, while channel-specific penalties allow different sparsity patterns across inputs and outputs.  Frequency-domain gain constraints and passivity constraints are naturally matrix-valued in this setting.  The same disk-kernel viewpoint applies, but the scalar measure is replaced by an operator-valued or matrix-valued measure when one interprets the full MIMO covariance structure.  In the finite RAF implementation, however, the optimization remains a structured regression over residue matrices.

\subsection{Limitations}
The framework has several limitations.  First, the quality of the finite dictionary matters: no sparse regularizer can identify modes outside the sampled region, and local refinement can only help when the active support provides a meaningful warm start.  Second, highly coherent disk-Vandermonde or input-filtered columns can lead to unstable support estimates even when prediction error is small.  Third, active atoms are not certified pole estimates; pole localization requires additional separation, excitation, signal-strength, and noise assumptions.  Fourth, the current theory is clearest for linear time-invariant impulse responses; nonlinear, LPV, and closed-loop extensions require additional assumptions on scheduling variables, excitation, and noise.  Fifth, hard priors can introduce bias.  The appropriate interpretation is therefore set-membership or physics-guided identification rather than purely data-driven model discovery.  These limitations are also opportunities: adaptive sampling, posterior pole distributions, robust conic formulations, and closed-loop instrumental-variable versions are natural extensions.

\section{Numerical Example}\label{sec:numerical}
We illustrate the method on a SISO system composed of two underdamped modes.  The system is excited by a low-bandwidth input, and 100 samples are collected.  The output is corrupted by additive white Gaussian noise with SNR approximately 30 dB.  Preliminary prediction-error fits reveal that the data are not sufficiently informative to identify the true impulse response uniquely: many stable models fit the measured output but extrapolate differently and recover different pole locations.

This is the intended use case for physics-informed RAF identification.  We impose pole-sector priors with radii in $[0.85,0.95]$ and angles in $[0.3,1]$, a BIBO budget $h_{\max}=55$, a DC-gain bound $G_{\max}=4$, and the frequency-response gain mask shown in Fig.~\ref{fig:gammabound}.  These priors need not all be used simultaneously; rather, they demonstrate how additional knowledge can be layered into the same convex program.

\begin{figure}
\centering
\includegraphics[width=0.99\linewidth]{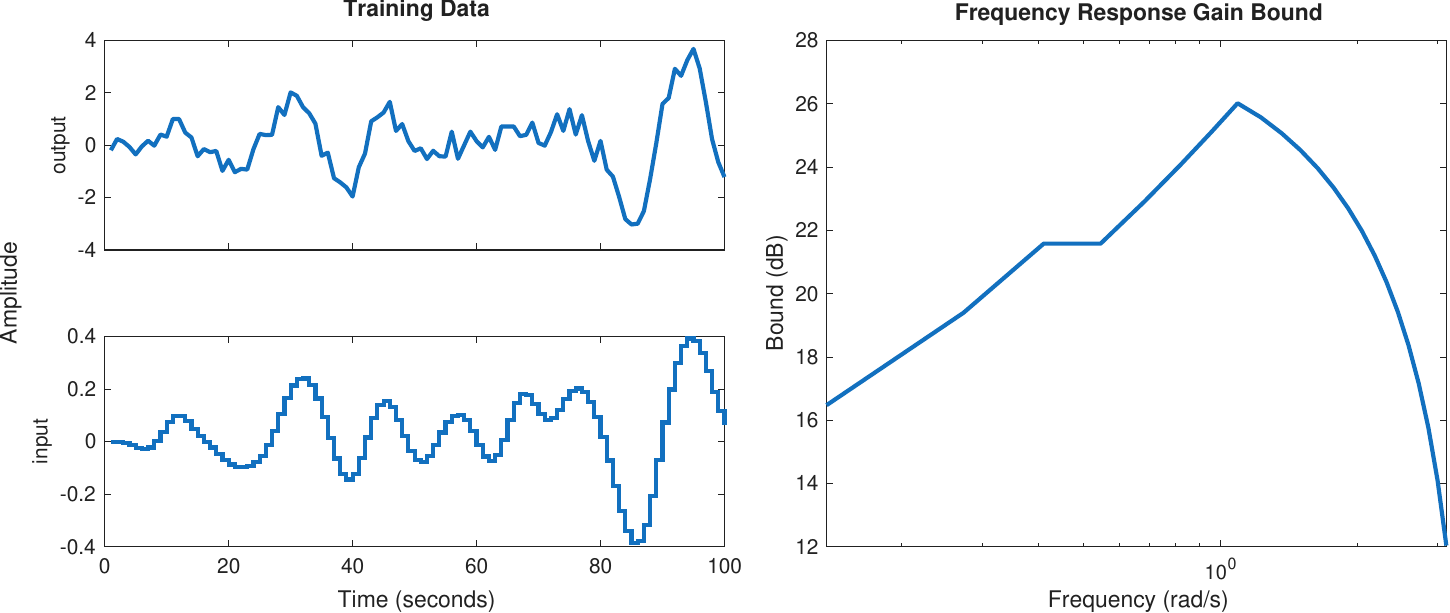}
\caption{Data and bound prior. Left: training input/output data. Right: frequency-response gain bound.}
\label{fig:gammabound}
\end{figure}

We compare four workflows.  PEM denotes conventional prediction-error modeling with a stability prior and order selected using Hankel-matrix rank heuristics \citep{Ljung:book99}.  KRM denotes tuned-correlated kernel regularized estimation \citep{chen2014maximum}, with a balanced-reduction variant labeled KRM-RED.  ANM denotes the randomized atomic norm method based on fully corrective Frank--Wolfe updates \citep{MillerATOM:2020}.  RAF denotes the proposed constrained randomized atomic feature estimator.

With only a stability prior, KRM gives the best output fit among the randomized feature methods, while ANM and RAF fit the data but may not recover the true impulse response.  This is expected because the data are poorly exciting and the dictionary is underconstrained.  Adding pole-sector information and sparsity improves the modal recovery for both ANM and RAF.  RAF then permits additional convex constraints, such as DC gain, BIBO budget, and frequency-response masks, which further improve the recovered impulse response and yield more informative candidate pole regions.  The results are summarized in Fig.~\ref{fig:results}.

\begin{figure*}
\centering
\includegraphics[width=0.99\linewidth]{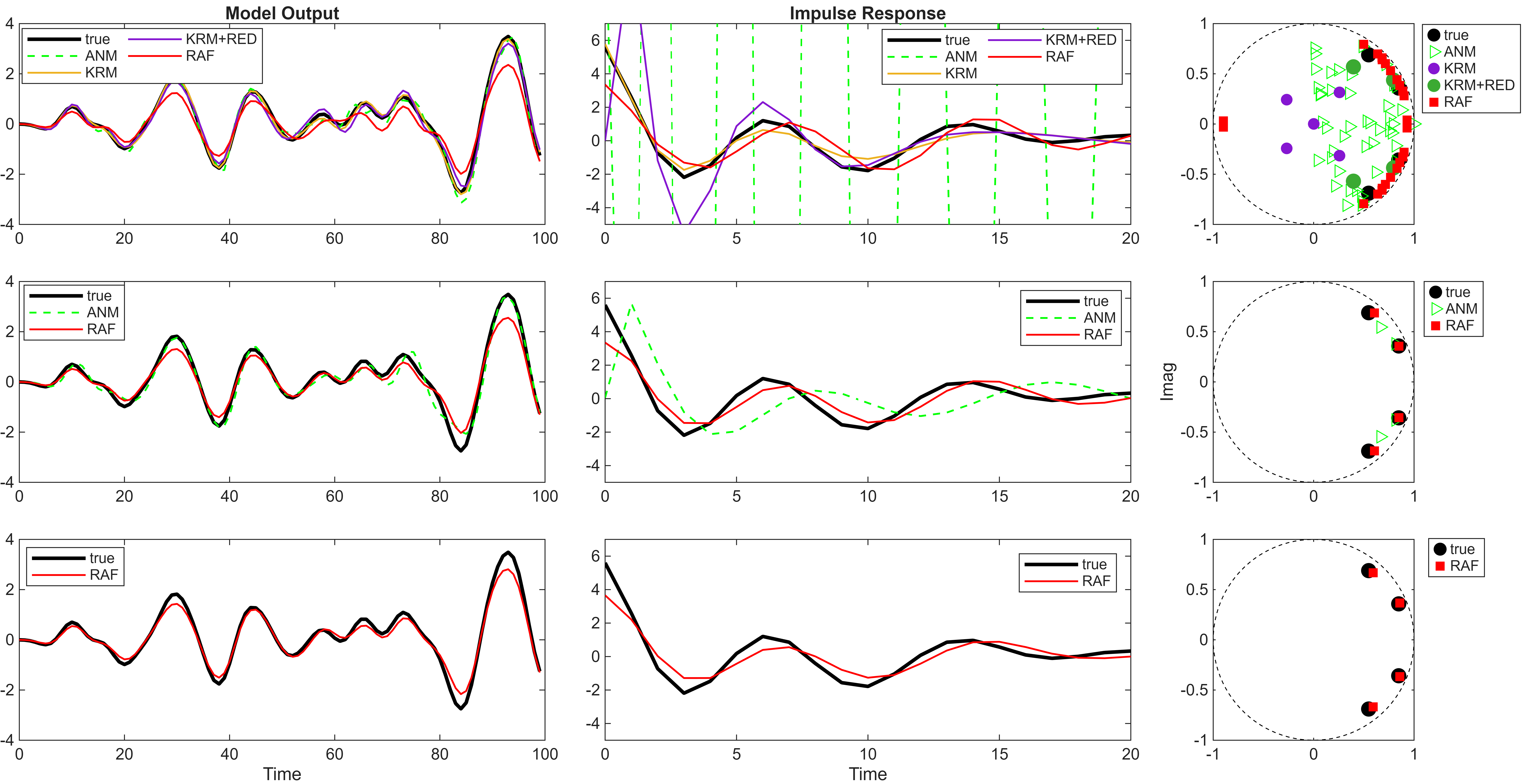}
\caption{Fitting results. Top row: stability prior only. Middle row: pole-sector constraints plus sparsity. Bottom row: RAF with all priors imposed.}
\label{fig:results}
\end{figure*}

The example is deliberately small and illustrative; it is not intended as a statistical benchmark or as evidence of guaranteed pole recovery.  It nevertheless highlights an important point.  RAF is not merely a randomized approximation to a kernel estimator.  It is a modeling language for combining modal dictionaries with convex side information.  When data are rich, the method behaves like a scalable sparse modal estimator.  When data are poor, the ability to impose hard engineering priors can be more important than the choice of loss function.

\section{Discussion}\label{sec:discussion}
The Disk--Bochner theorem above explains the kernel generated by RAF atoms without changing the practical algorithm.  The finite RAF model uses $A=\diag(p_1,\ldots,p_M)$, which is normal, and therefore its kernel is generated by an atomic positive disk measure.  For arbitrary kernels, the radius defect gives a contractive canonical shift; the additional subnormality condition is exactly what permits a scalar positive measure over pole locations.

The Nevanlinna--Pick connection is a transfer-function statement in the disk-analytic delay variable $\zeta=z^{-1}$.  Gain normalization maps a bounded transfer function $H(\zeta)$ to a Schur function $F(\zeta)=H(\zeta)/\gamma$, and the Pick matrix together with the LFT parametrization describes all normalized transfer functions satisfying the interpolation and gain constraints.  Pole-radius normalization is a separate structural operation: it acts on the pole coordinate $p$ and yields the radius defect $\rho^2K(s,t)-K(s+1,t+1)\succeq0$ for the pole-moment kernel.  RAF combines these two ingredients by restricting the modal search space through sampled poles while certifying gain or interpolation constraints in the appropriate transfer-function coordinate.

The sparse-recovery discussion is intentionally conditional.  Existing atomic-norm and gridless sparse-recovery theory for exponential models provides important intuition, especially for well-separated angular frequencies on the unit circle \citep{Bhaskar2012Atomic,Chi2015Compressive,Tang2013b}.  The present RAF setting is different: it fixes a finite set of candidate poles, uses the input-filtered design $X_{t,m}=(\phi_m*u)(t)$, and may include additional priors that deliberately bias the model.  Rigorous localization would require radial-angular separation in the disk geometry, sufficiently informative inputs, noise control, signal-strength conditions, and favorable conditioning or incoherence of the realized design.  In applications, this suggests using sampling designs that avoid clusters in pseudo-hyperbolic geometry, monitoring column coherence after convolution with the measured input, and treating active atoms as warm starts rather than certified continuous pole estimates.

\section{Conclusions}\label{sec:concl}
Randomized atomic features provide a practical bridge between kernel methods and explicit modal dictionaries for physics-informed system identification.  By sampling stable poles in the disk and fitting sparse residues with convex constraints, RAF models retain modal interpretability while supporting stability margins, frequency masks, DC-gain bounds, monotonicity, passivity, relative degree, and time-domain error envelopes.  Positive disk measures generate PSD kernels with a radius-dependent shift defect and stable RKHS embeddings; the corresponding scalar disk representation for arbitrary kernels is characterized by subnormality of the canonical shift.  The resulting framework connects randomized features, atomic gauges, stable kernels, and Schur/NP--LFT set-membership in a way that is directly useful for constrained identification.  Future work will extend the statistical theory for disk-Vandermonde designs, develop adaptive sampling laws, and apply the framework to MIMO, closed-loop, and nonlinear LPV identification.

\appendix
\section{Additional Operator Details}\label{app:operator}
The canonical construction in Definition~\ref{def:canonical} is unique up to unitary equivalence.  If another cyclic pair $(A,b)$ satisfies $K(s,t)=\langle A^sb,A^tb\rangle$ and $\overline{\operatorname{span}}\{A^tb:t\ge0\}$ is the whole state space, then the map $e_t\mapsto A^tb$ extends to a unitary intertwining the two realizations.  Therefore subnormality of the canonical shift is an intrinsic property of the kernel, not an artifact of a chosen realization.

In finite dimensions, normality is equivalent to unitary diagonalizability.  If $A=V\diag(p_1,\ldots,p_n)V^*$ and $b$ is any vector, then
\begin{equation}
\langle A^sb,A^tb\rangle=\sum_{j=1}^n |(V^*b)_j|^2p_j^s\bar p_j^{\,t}.
\end{equation}
This is an atomic disk measure.  A Schur triangular form $A=QTQ^*$ is not enough: the strictly upper triangular part of $T$ produces polynomial factors multiplying $p^t$, which cannot generally be absorbed into positive scalar atoms.

\section{Implementation Notes}\label{app:implementation}
For complex-valued dictionaries, one may either work directly with complex conic solvers or split all variables into real and imaginary parts.  The epigraph constraint $a_m\ge |c_m|$ becomes
\begin{equation}
\|(\Re c_m,\Im c_m)\|_2\le a_m.
\end{equation}
Group penalties over conjugate pairs are handled similarly.  When real-valued impulse responses are required, poles are sampled in conjugate pairs and the residues are constrained to be conjugates, or an equivalent real sinusoidal basis is used.  The parallel realization can be ill-conditioned when many poles are close; in that case column normalization and post-solve residue rescaling are recommended before applying sparsity penalties.

\bibliography{ifacconf}

\end{document}